\documentclass[12pt]{article}
 
\usepackage{verbatim}
\usepackage{fullpage}
\usepackage{amsthm}
\usepackage{amsmath} 
\usepackage{amssymb}
\usepackage{bm}
\usepackage{enumerate}  
\usepackage{braket}
\usepackage{dsfont}

\swapnumbers
\newtheorem{theorem}{Theorem}[section]
\newtheorem{lemma}[theorem]{Lemma}
\newtheorem{cor}[theorem]{Corollary}

{\theoremstyle{remark}
\newtheorem{example}[theorem]{Example}
}

\newcommand{\N}{\mathds{N}}
\newcommand{\Q}{\mathds{Q}}
\newcommand{\esupp}[1]{\Theta_{#1}}
\newcommand{\norm}[1]{\left\lVert#1\right\rVert}

\allowdisplaybreaks

\title{Pretty Good State Transfer of\\Multiple Qubit States on Paths}
\author{Christopher M. van Bommel\footnote{Partially supported by a Canada Graduate Scholarship (Doctoral) from the Natural Sciences and Engineering Research Council of Canada}\\ Department of Combinatorics and Optimization\\ University of Waterloo, ON, Canada\\ \texttt{cvanbomm@uwaterloo.ca}}
\date{\today}

\begin{document}

\maketitle

\begin{abstract}
We discuss pretty good state transfer of multiple qubit states and provide a model for considering state transfer of arbitrary states on unmodulated XX-type spin chains.  We then provide families of paths and initial states for which we can determine whether there is pretty good state transfer based on the eigenvalue support of the initial state.

\vspace{5pt} \noindent Keywords: quantum walks, state transfer, graph eigenvalues 

\vspace{5pt} \noindent Mathematical Subject Classification: 81P68, 15A16, 05C50
\end{abstract}

\section{Introduction}

In quantum information processing, a key requirement is the ability to transfer quantum states from one location to another.  Perhaps an obvious method of doing so is via a series of SWAP gates; however, such a procedure would require a great deal of control over the system and would be highly prone to errors, as analyzed by Petrosyan, Nikolopoulos, and Lambropoulos~\cite{PNL10}.  Instead, one could take advantage of the natural propagation of the system as time passes to transfer quantum states.  The protocol for quantum communication through unmeasured and unmodulated spin chains was presented by Bose~\cite{B03}, and led to the interpretation of quantum channels implemented by spin chains as wires for transmission of states.

In the ideal scenario, the fidelity of this state transfer is equal to one, and we say we have \emph{perfect state transfer}.  This concept of perfect state transfer was introduced by Christandl et al.~\cite{CDDEKL05}, who showed that perfect state transfer on uniformly coupled spin chains is only possible for chains of two or three qubits.  If non-uniform coupling schemes are considered, then perfect state transfer can be realized on spin chains of arbitrary length, as demonstrated by Christandl et al.~\cite{CDDEKL05} and Yung and Bose~\cite{YB05}, however, engineering such a scheme would be highly difficult in practice.

Hence, it would be desirable to demonstrate the achievement of quantum state transfer in spin chains where there is little variation among the coupling strengths.  Motivated by this desire, Karbach and Stolze~\cite{KS05} demonstrated perfect state transfer using weakly varying coupling configurations.  Multiple authors have also considered modifying the coupling strengths only of the couplings near the ends of the spin chain~\cite{Banchi2013,BACVV11,BBVB11,Wojcik2005}.  Other alternatives considered include iterative measurement procedures \cite{Bayat14,BB05} and initializing the channel qubits to a specific state~\cite{BBBV11}.

On the other hand, it is worth investigating whether relaxing the requirement on the fidelity would lead to better implementation.  Such a question is motivated by arguments that these implementations are too demanding compared to the level of fidelity required for most tasks in quantum information processing~\cite{Zenchuk2012}.  Hence, pretty good state transfer has been introduced by multiple authors, including Vinet and Zhedanov~\cite{VZ12} (under the term almost perfect state transfer) and Godsil~\cite{Godsil2012}, which requires there to be times where the fidelity of transfer is arbitrarily close to one.  Godsil, Kirkland, Servini, and Smith~\cite{GKSS12} considered pretty good state transfer on unmodulated XX-type spin chains, or from a graph theoretic standpoint unweighted paths with respect to their adjacency matrices, and proved the following.  We let $P_n$ denote a path on $n$ vertices, and assume without loss of generality that the vertices of $P_n$ are labelled 1 to $n$ such that vertices with consecutive labels are adjacent.

\begin{theorem} \cite{GKSS12}
Pretty good state transfer occurs between the end vertices of $P_n$ if and only if $n = p - 1$, $2p - 1$, where $p$ is a prime, or $n = 2^m - 1$.  Moreover, when pretty good state transfer occurs between the end vertices of $P_n$, then it occurs between vertices $a$ and $n + 1 - a$ for all $a \neq (n + 1) / 2$.
\end{theorem}

Coutinho, Guo, and van Bommel~\cite{CGvB17} determined an infinite family of paths which exhibit pretty good state transfer between internal vertices but not between the end vertices, and van Bommel~\cite{vB19} completed the characterization by showing there were no additional examples, leading to the following result.

\begin{theorem} \label{thm-pgst-path} \cite{vB19}
There is pretty good state transfer on $P_n$ between vertices $a$ and $b$ if and only if $a + b = n + 1$ and either:
\begin{enumerate}[a)]
\item $n = 2^t - 1$, where $t$ is a positive integer; or, 
\item $n = 2^t p - 1$, where $t$ is a nonnegative integer and $p$ is an odd prime, and $a$ is a multiple of $2^{t - 1}$.
\end{enumerate}
\end{theorem}

Thus far, our discussion has been limited to single-particle qubit states.  However, the usefulness of the implementation of uniformly coupled XX-type spin chains also depends on whether many particle qubit states, and more importantly, entangled states, can be transferred with arbitrarily high fidelity through the chain.    Albanese, Christandl, Datta, and Ekert~\cite{ACDE04} considered implementing a mirror inversion of the spin chain with respect to its centre, and demonstrated two families of coupling strengths exhibiting this property.  Karbach and Stolze~\cite{KS05} generalized these two families to an infinite number of cases exhibiting mirror inversion.  The key observation used in these investigations is that perfect state transfer of the end vertices extends to perfect state transfer of any multiple-particle state.

Sousa and Omar~\cite{SO14} extend the definition of pretty good state transfer to multiple-particle qubit single-excitation states as follows.  They say there is pretty good state transfer of the $m$-qubit state $\ket{\psi_{in}}$ if for any $\epsilon > 0$, there is a time $t > 0$ such that
\[
| \bra{n + 1 - j} U(t) \ket{j} | > 1 - \epsilon, \qquad j = 1, \ldots, m.
\]
Under this definition, they conclude the following result.

\begin{theorem} \cite{SO14}
If the multi-qubit input state $\ket{\psi_{in}}$ is restricted to the single-excitation manifold, then there is pretty good state transfer of $\ket{\psi_{in}}$ if%
\footnote{The authors claim this result to be ``if and only if'', but do not provide a proof of the other direction.}%
$n = p - 1$, $2p - 1$, or $2^k - 1$, where $p$ is a prime and $k \in \N$.
\end{theorem}

This definition of pretty good state transfer appears to be more strict than required.  For one, it is restricted to a path, and moreover in Example~\ref{ex:P11} we will demonstrate a pair of vertices on a path which allow almost mirror inversion but which individually do not admit pretty good state transfer.  Hence we propose defining that a graph has \emph{pretty good state transfer} between states $\ket{\mathbf{v}}$ and $\ket{\mathbf{w}}$ if, for any $\epsilon > 0$, there is a time $t >0$ such that
\[
| \bra{\mathbf{w}} U(t) \ket{\mathbf{v}} | > 1 - \epsilon.
\]
We then provide families of paths and initial states for which we can determine whether there is pretty good state transfer based on the eigenvalue support of the initial state.  The results of this paper can also be found in the author's Ph.D.\ dissertation~\cite{PHD}.

\section{State Transfer Model}

In this work, we will exclusively consider modelling interacting qubits of a spin chain, or path.  We consider the model described with an XX-type Hamiltonian with uniform coupling strengths, given by
\[
H = \frac{1}{2} \sum_{j = 1}^{n - 1} J (\sigma_j^x \sigma_{j + 1}^x + \sigma_j^y \sigma_{j + 1}^y),
\]
where $\sigma_j^x$, $\sigma_j^y$, and $\sigma_j^z$ are Pauli matrices acting at position $j$.  We may assume without loss of generality that $J = 1$.  Using the Jordan--Wigner transformation~\cite{JW}, we obtain that when the initial state is restricted to the single excitation space, the evolution of the process is governed by $U(t) = e^{i A t}$, where $A$ is the adjacency matrix of the path.
We note that the Hamiltonian preserves the number of excitations, so as the process evolves, all states will be in the single excitation space.

We consider the following generalization of single-qubit transfer.  Suppose the state sender $S$, having access to one end of the chain, wants to send the $m$-qubit state 
\[
\ket{\mathbf{x}_S} = \sum_{j = 1}^m \beta_j \ket{j}, \qquad \sum_{j = 1}^m |\beta_j|^2 = 1,
\]
where $\ket{j}$ corresponds to the state with only the $j$th qubit in the excited state.  The goal is to send this state to the receiver $R$, having access to the other end of the chain.  So a generic state of the system has the form
\[
\ket{\mathbf{x}} = \sum_{j = 1}^m \beta_j \ket{j}, \qquad \sum_{j = 1}^m |\beta_j|^2 = 1,
\]
where qubits 1 to $m$ correspond to the qubits from which the state is sent, qubits $n - m + 1$ to $n$ correspond to the quibts on which the state is to be received, and qubits $m + 1$ to $n - m$ correspond to the remainder of the spin chain.  We initialize the spin chain to the state $\ket{\mathbf{0}}$, and then the sender places state $\ket{\mathbf{x}_S}$ on the first $m$ qubits, resulting in the system having state $\ket{\mathbf{v}} = \ket{\mathbf{x}_S} \otimes \ket{\mathbf{0}}$.  We then wish for the receiver to receive the mirrored state on the other end of the spin chain, which allows us to take advantage of the symmetry of the spin chain, hence the state of the system we seek is 
\[
\ket{\mathbf{v}^\sigma} := \sum_{j = 1}^n \beta_{n + 1 - j} \ket{j}.
\]
Hence, to compare our desired received state to the output state, we compute the quantum fidelity given by
\begin{align*}
\bra{\mathbf{v}^\sigma} \left( U(\tau) \ket{\mathbf{v}} \bra{\mathbf{v}} U(-\tau) \right) \ket{\mathbf{v}^\sigma} &= \bra{\mathbf{v}^\sigma} U(\tau) \ket{\mathbf{v}} \left( \bra{\mathbf{v}^\sigma} U(\tau) \ket{\mathbf{v}} \right)^\dagger = |\bra{\mathbf{v}^\sigma} U(\tau) \ket{\mathbf{v}}|^2.
\end{align*}

\section{Algebra, Graphs, and Number Theory}

In this section, we provide an overview of the basic definitions and theory from algebraic graph theory and number theory that will be used throughout this paper.    We model a spin chain of $n$ interacting qubits by the graph of a path of $n$ vertices, denoted $P_n$, with the vertices labelled from 1 to $n$ corresponding to qubits and the edges $\{j, j + 1\}$, $1 \le j < n$ corresponding to their interactions.  If $M$ is a symmetric matrix with $d$ distinct eigenvalues $\theta_1 > \theta_2 > \cdots > \theta_d$, then the spectral decomposition of $M$ is
\[
M = \sum_{j = 1}^d \theta_j E_j,
\]
where $E_j$ denotes the orthogonal projection onto the eigenspace corresponding to $\theta_j$.  For a state $\ket{\mathbf{v}}$, the \emph{eigenvalue support} of $\ket{\mathbf{v}}$ is the following subset of the eigenvalues:
\[
\esupp{\ket{\mathbf{v}}} = \{ \theta_j : E_j \ket{\mathbf{v}} \neq \ket{\mathbf{0}}\}.
\]
We say that states $\ket{\mathbf{v}}$ and $\ket{\mathbf{w}}$ are \emph{cospectral} if for each idempotent $E_j$ in the spectral decomposition of $X$, we have $\bra{\mathbf{v}} E_j \ket{\mathbf{v}} = \bra{\mathbf{w}} E_j \ket{\mathbf{w}}$, parallel if for each eigenvalue $\theta_j$, the vectors $E_j \ket{\mathbf{v}}$ and $E_j \ket{\mathbf{w}}$ are parallel, and \emph{strongly cospectral} if for each eigenvalue $\theta_j$, there exists a $\gamma_j$ such that $|\gamma_j| = 1$ and $E_j \ket{\mathbf{v}} = \gamma_j E_j \ket{\mathbf{w}}$; these definitions expand the notions of cospectrality, parallel, and strong cospectrality to multiple qubit states (for strongly cospectral vertices $a$ and $b$, we can make the stronger statement that $E_j \ket{a} = \pm E_j \ket{b}$).  As we have with vertices, strongly cospectral is equivalent to cospectral and parallel, as demonstrated in the following lemma.

\begin{lemma}
Two states $\ket{\mathbf{v}}$ and $\ket{\mathbf{w}}$ in $X$ are strongly cospectral if and only if they are parallel and cospectral.
\end{lemma}

\begin{proof}
Suppose states $\ket{\mathbf{v}}$ and $\ket{\mathbf{w}}$ are strongly cospectral.  It immediately follows from the definition that $\ket{\mathbf{v}}$ and $\ket{\mathbf{w}}$ are parallel.  We verify that $\ket{\mathbf{v}}$ and $\ket{\mathbf{w}}$ are cospectral by observing
\[
\bra{\mathbf{v}} E_j \ket{\mathbf{v}} = \gamma_j \bra{\mathbf{v}} E_j \ket{\mathbf{w}} = \gamma_j \overline{\gamma}_j \bra{\mathbf{w}} E_j \ket{\mathbf{w}} = \bra{\mathbf{w}} E_j \ket{\mathbf{w}},
\]
using the fact that the idempotents are real symmetric.

Conversely, suppose $\ket{\mathbf{v}}$ and $\ket{\mathbf{w}}$ are parallel and cospectral.  It follows that
\[
\bra{\mathbf{w}} E_j \ket{\mathbf{v}} = c \bra{\mathbf{w}} E_j \ket{\mathbf{w}} = c \bra{\mathbf{v}} E_j \ket{\mathbf{v}} = c \overline{c} \bra{\mathbf{w}} E_j \ket{\mathbf{v}} = |c|^2 \bra{\mathbf{w}} E_j \ket{\mathbf{v}}
\]
again using the fact that idempotents are symmetric.  Hence it follows that if $\bra{\mathbf{w}} E_j \ket{\mathbf{v}}$ is nonzero, then $|c|^2 = 1$, so $c = \gamma_j, |\gamma_j| = 1$ as required.

Now consider the situation when $\bra{\mathbf{w}} E_j \ket{\mathbf{v}} = 0$.  Since $\ket{\mathbf{v}}$ and $\ket{\mathbf{w}}$ are parallel, it follows that $\bra{\mathbf{v}} E_j \ket{\mathbf{v}} = \bra{\mathbf{w}} E_j \ket{\mathbf{w}} = 0$.  We then observe that 
\[
\norm{E_j \ket{\mathbf{v}}}^2 = \bra{\mathbf{v}} E_j^\dagger E_j \ket{\mathbf{v}} = \bra{\mathbf{v}} E_j \ket{\mathbf{v}} = 0,
\]
using the fact that $E_j$ is real symmetric and idempotent.  Similarly, $\norm{E_j \ket{\mathbf{w}}}^2 = 0$.  Hence $E_j \ket{\mathbf{v}} = E_j \ket{\mathbf{w}} = \ket{\mathbf{0}}$ as required.
\end{proof}

A connection between pretty good state transfer and strongly cospectral vertices was first observed by Dave Witte Morris (as cited in~\cite{Godsil2012}).  We will prove the analogous result for arbitrary states.

\begin{lemma} \cite[private communication with Morris]{Godsil2012}
Let $X$ be a graph and let $a$ and $b$ be vertices of $X$.  If there is pretty good state transfer between $\ket{a}$ and $\ket{b}$, then $a$ and $b$ are strongly cospectral.
\end{lemma}

\begin{lemma}
Let $X$ be a graph and let $\ket{\mathbf{v}}$ and $\ket{\mathbf{w}}$ be states of $X$.  If there is pretty good state transfer between $\ket{\mathbf{v}}$ and $\ket{\mathbf{w}}$, then $\ket{\mathbf{v}}$ and $\ket{\mathbf{w}}$ are strongly cospectral.
\end{lemma}

\begin{proof}
By definition, if we have pretty good state transfer from $\ket{\mathbf{v}}$ to $\ket{\mathbf{w}}$ in $X$, then there exists a sequence of times $\{t_k\}$ such that
\[
\lim_{k \to \infty} |\bra{\mathbf{w}} U(t_k) \ket{\mathbf{v}}| = 1.
\]
We calculate the following:
\begin{align*}
1 &= \lim_{k \to \infty} |\bra{\mathbf{w}} U(t_k) \ket{\mathbf{v}}| \\
&\le \sum_{\theta} | \bra{\mathbf{w}} E_\theta \ket{\mathbf{v}} | \\
&\le \sum_{\theta} \sqrt{\bra{\mathbf{v}} E_\theta \ket{\mathbf{v}}} \sqrt{\bra{\mathbf{w}} E_\theta \ket{\mathbf{w}}} \\
&\le \sqrt{ \sum_\theta \bra{\mathbf{v}} E_\theta \ket{\mathbf{v}} \sum_\theta \bra{\mathbf{w}} E_\theta \ket{\mathbf{w}} } = 1.
\end{align*}
The first inequality is an application of the triangle inequality.  The second and third inequalities are applications of Cauchy-Schwarz, where we take $u = E_\theta \ket{\mathbf{v}}$ and $v = E_\theta \ket{\mathbf{w}}$ and $u = (\sqrt{\bra{\mathbf{v}} E_\theta \ket{\mathbf{v}}})_\theta$ and $v = (\sqrt{\bra{\mathbf{w}} E_\theta \ket{\mathbf{w}}})_\theta$ respectively. The last inequality follows from the spectral decomposition and the definition of $\ket{\mathbf{v}}, \ket{\mathbf{w}}$.

Therefore, all inequalities must hold with equality.  In particular, the second inequality implies $\ket{\mathbf{v}}$ and $\ket{\mathbf{w}}$ are parallel, and the third inequality implies $\ket{\mathbf{v}}$ and $\ket{\mathbf{w}}$ are cospectral, which completes the proof.
\end{proof}

The spectrum of the adjacency matrix of $P_n$ (see \cite{BH12} for example), is
\[
\theta_j = 2 \cos \frac{\pi j}{n + 1}, \quad 1 \le j \le n,
\]
and the eigenvector corresponding to $\theta_j$ is given by
\[
\ket{\bm{\beta}} = \sum_{k = 1}^n \beta_k \ket{k}, \quad \beta_k = \sin \frac{k \pi j}{n + 1}.
\]
Hence the spectral idempotent $E_j$ corresponding to $\theta_j$ is such that
\[
\bra{k} (E_j) \ket{\ell} = \frac{2}{n + 1} \sin \left( \frac{k j \pi}{n + 1} \right) \sin \left( \frac{\ell j \pi}{n + 1} \right).
\]
Thus, we are able to demonstrate the following.

\begin{lemma} \label{lem:path-sc}
Let $\ket{\mathbf{v}}$ be a state of a path $P_n$.  Then $\ket{\mathbf{v}}$ and $\ket{\mathbf{v}^\sigma}$ are strongly cospectral.  Moreover, $E_j \ket{\mathbf{v}} = (-1)^{j + 1} E_j \ket{\mathbf{v}^\sigma}$.
\end{lemma}

\begin{proof}
First consider $\ket{\mathbf{v}} = \ket{x}$ and $\ket{\mathbf{v}^\sigma} = \ket{n + 1 - x}$.  For a fixed $y$, we have
\begin{align*}
\bra{y} E_j \ket{x} %
&= \frac{2}{n + 1} \sin \left( \frac {x j \pi}{n + 1} \right) \sin \left( \frac{y j \pi}{n + 1} \right) \\
\bra{y} E_r \ket{n + 1 - x} %
&= \frac{2}{n + 1} \sin \left( \frac{(n + 1 - x) j \pi}{n + 1} \right) \sin \left( \frac{y j \pi}{n + 1} \right) \\
&= \frac{2 }{n + 1} \left( \sin(j \pi) \cos \left( \frac{x j \pi}{n + 1} \right) - \cos(j \pi) \sin \left( \frac{x j \pi}{n + 1} \right) \right)  \sin \left( \frac{y j \pi}{n + 1} \right) \\
&= (-1)^{j + 1} \bra{y} E_j \ket{x},
\end{align*}
from which it follows that $E_j \ket{x} = (-1)^{j + 1} E_j \ket{n + 1 - x}$.  It follows that for a given $j$, we have
\begin{align*}
E_j \ket{\mathbf{v}} %
&= E_j \sum_{x \in V(P_n)} \beta_x \ket{x} \\
&= \sum_{x \in V(P_n)} \beta_x E_j \ket{x} \\
&= \sum_{x \in V(P_n)} \beta_x (-1)^{j + 1} E_j \ket{n + 1 - x} \\
&= (-1)^{j + 1} E_j \sum_{x \in V(P_n)} \beta_x \ket{n + 1 - x} \\
&= (-1)^{j + 1} E_j \ket{\mathbf{v}^\sigma},
\end{align*}
as desired.
\end{proof}

Finally, we present Kronecker's Theorem as a key tool that will be used throughout.

\begin{theorem}[Kronecker, see~\cite{LZ82}] \label{KT}
Let $\theta_0, \ldots, \theta_d$ and $\zeta_0, \ldots, \zeta_d$ be arbitrary real numbers.  For an arbitrarily small $\epsilon$, the system of inequalities
\[
| \theta_r y - \zeta_r | < \epsilon \pmod{2 \pi}, \quad (r = 0, \ldots, d),
\]
admits a solution for $y$ if and only if, for integers $\ell_0, \ldots, \ell_d$ such that
\[
\ell_0 \theta_0 + \cdots + \ell_d \theta_d = 0,
\]
then
\[
\ell_0 \zeta_0 + \cdots + \ell_d \zeta_d \equiv 0 \pmod{2 \pi}.
\]
\end{theorem}

\section{Main Results}

We begin this section by demonstrating that in $P_{11}$, the state $\frac{1}{\sqrt{2}} \left( \ket{1} + \ket{3} \right)$ admits almost mirror inversion but neither $\ket{1}$ nor $\ket{3}$ admit pretty good state transfer, justifying expanding the definition of pretty good state transfer of multiple qubit states.

\begin{example} \label{ex:P11}
By Theorem~\ref{thm-pgst-path}, we see that neither $\ket{1}$ nor $\ket{3}$ admit pretty good state transfer on $P_{11}$.  It remains to verify there is almost mirror inversion on $P_{11}$ of $\ket{\mathbf{v}} = \frac{1}{\sqrt{2}} \left( \ket{1} + \ket{3} \right)$ to $\ket{\mathbf{v}^\sigma} = \frac{1}{\sqrt{2}} \left( \ket{11} + \ket{9} \right)$. We first demonstrate that $\theta_6 = \ket{\mathbf{0}} \notin \Theta_{\ket{\mathbf{v}}}$.  We have that
\begin{align*}
\bra{x} E_6 \ket{\mathbf{v}} &= \frac{1}{\sqrt{2}} \left( \bra{x} E_6 \ket{1} + \bra{x} E_6 \ket{3} \right) \\ &= \frac{1}{\sqrt{2}} \left( \frac{1}{6} \sin \left( \frac{\pi x}{2} \right) \sin \left( \frac{\pi}{2} \right) + \frac{1}{6} \sin \left( \frac{\pi x}{2} \right) \sin \left( \frac{3 \pi}{2} \right) \right) = 0,
\end{align*}
and so $E_6 \ket{\mathbf{v}} = \ket{\mathbf{0}}$ as claimed.

Now, we wish to show that
\[
| \bra{\mathbf{v}^\sigma} U(\tau) \ket{\mathbf{v}} | > 1 - \epsilon,
\]
which, using spectral decomposition, is equivalent to
\[
\left| \bra{\mathbf{v}^\sigma} \left( \sum_{\theta_j \in \Theta_{\ket{\mathbf{v}}}} \exp(i \tau \theta_j) E_j \right) \ket{\mathbf{v}} \right| > 1 - \epsilon,
\]
and by linearity, we obtain
\[
\left| \sum_{\theta_j \in \Theta_{\ket{\mathbf{v}}}} \bra{\mathbf{v}^\sigma} \exp(i \tau \theta_j) E_j \ket{\mathbf{v}} \right| > 1 - \epsilon.
\]
We then observe that $E_j \ket{\mathbf{v}} = (-1)^{j + 1} E_j \ket{\mathbf{v}^\sigma}$ and obtain
\[
\left| \sum_{\theta_j \in \Theta_{\ket{\mathbf{v}}}} \bra{\mathbf{v}^\sigma} (-1)^{j + 1} \exp(i \tau \theta_j) E_j \ket{\mathbf{v}^\sigma} \right| > 1 - \epsilon
\]
from which it follows that we desire
\[
|\tau \theta_j - \sigma_j \pi - \delta| < \epsilon' \pmod{2\pi}, \quad (j : \theta_j \in \Theta_{\ket{\mathbf{v}}}),
\]
where $\sigma_r$ is even if $r$ is odd and odd if $r$ is even, and $\delta$ is some fixed value.

If we let $\delta = 0$, then the inequalities corresponding to $\theta_j$ and $\theta_{11 - j}$ differ only by a factor of $-1$, so it suffices to consider the system
\begingroup
\addtolength{\jot}{1em}
\begin{align*}
\left| \frac{\sqrt{6} + \sqrt{2}}{2} \tau \right| &< \epsilon' \pmod{2\pi}, \\
\left| \sqrt{3} \tau + \pi \right|  &< \epsilon' \pmod{2\pi}, \\
\left| \sqrt{2} \tau \right|  &< \epsilon' \pmod{2\pi}, \\
\left| \tau + \pi \right|  &< \epsilon' \pmod{2\pi}, \\
\left| \frac{\sqrt{6} - \sqrt{2}}{2} \tau \right| &< \epsilon' \pmod{2\pi}.
\end{align*}
\endgroup
In order to apply Kronecker's Theorem (Theorem~\ref{KT}), we need to show that for every integer solution to the equation
\[
\frac{\sqrt{6} + \sqrt{2}}{2} \ell_1 + \sqrt{3} \ell_2 + \sqrt{2} \ell_3 + \ell_4 + \frac{\sqrt{6} - \sqrt{2}}{2} \ell_5 = 0,
\]
we have
\[
0 (\ell_1 + \ell_3 + \ell_5) + \pi (\ell_2 + \ell_4) \equiv 0 \pmod{2\pi}.
\]
Since the only integer solutions have that $\ell_2 = \ell_4 = 0$, the above equation is satisfied.  Hence, we can apply Kronecker's Theorem, which verifies that we have almost mirror inversion between states $\ket{\mathbf{v}} = \frac{1}{\sqrt{2}} \left( \ket{1} + \ket{3} \right)$ and $\ket{\mathbf{v}^\sigma} = \frac{1}{\sqrt{2}} \left( \ket{11} + \ket{9} \right)$. \hfill $\Diamond$
\end{example}

We now proceed to develop the tools required to prove our main result.  We first consider an extension of the following result due to Banchi et al.~\cite{BCGS16}.

\begin{theorem} \cite{BCGS16}
Let $a$ and $b$ be vertices of a graph $X$.  Then pretty good state transfer occurs between $a$ and $b$ if and only if both conditions below are satisfied.
\begin{enumerate}[a)]
\item Vertices $a$ and $b$ are strongly cospectral, in which case let $\zeta_i = (1 - \sigma_i) / 2$.
\item If there is a set of integers $\{\ell_j\}$ such that
\[
\sum_{\theta_j \in \esupp{a}} \ell_j \theta_j = 0 \mbox{ and } \sum_{\theta_j \in \esupp{a}} \ell_j \zeta_j \mbox{ is odd},
\]
then
\[
\sum_{\theta_j \in \esupp{a}} \ell_j \neq 0.
\]
\end{enumerate}
\end{theorem}

\begin{theorem} \label{thm:KT-path}
Let $\ket{\mathbf{v}}$ be a state of a path $P_n$ and let $\zeta_j = (1 + (-1)^j) / 2$.  Then there is pretty good state transfer between states $\ket{\mathbf{v}}$ and $\ket{\mathbf{v}^\sigma}$ if and only if for every set of integers $\{\ell_j\}$ such that
\[
\sum_{\theta_j \in \esupp{\ket{\mathbf{v}}}} \ell_j \theta_j = 0 \mbox{ and } \sum_{\theta_j \in \esupp{\ket{\mathbf{v}}}} \ell_j \zeta_j \mbox{ is odd},
\]
then
\[
\sum_{\theta_j \in \esupp{\ket{\mathbf{v}}}} \ell_j \neq 0.
\]
\end{theorem}

\begin{proof}
First suppose the condition on the sets of integers $\{\ell_j\}$ is satisfied.  We consider the system of inequalities
\[
| \theta_j \tau - (\delta + \zeta_j \pi)| < \epsilon \pmod{2 \pi}, \qquad (\theta_j \in \esupp{\ket{\mathbf{v}}}).
\]
Let $\{\ell_j\}$ be a set of integers such that
\[
\sum_{\theta_j \in \esupp{\ket{\mathbf{v}}}} \ell_j \theta_j = 0.
\]
Then we desire
\begin{equation} \tag{$\dagger$}
\sum_{\theta_j \in \esupp{\ket{\mathbf{v}}}} \ell_j (\delta + \zeta_j \pi) \equiv 0 \pmod{2 \pi}.
\end{equation}

We need to show that there exists a $\delta$ such that the above equation is true for all sets $\{\ell_j\}$.  If $\sum \ell_j \zeta_j$ is even for every set of integers $\{\ell_j\}$, then we may choose $\delta = 0$.  Otherwise, suppose for some set of integers $\{\ell_j\}$ that $\sum \ell_j \zeta_j$ is odd, in which case let $\alpha := \sum \ell_j \neq 0$, and let $\delta$ be such that ($\dagger$) is satisfied for this set.  Suppose $\{\ell'_j\}$ is also a set of integers such that
\[
\sum_{\theta_j \in \esupp{\ket{\mathbf{v}}}} \ell'_j \theta_j = 0.
\]
Let $\alpha' := \sum \ell'_j$, let $r, s$ be nonnegative integers such that $2^r$ is the largest power of 2 that divides $\alpha$ and $2^s$ is the largest power of 2 that divides $\alpha'$, and let $t = \min\{r, s\}$.  Take $\delta = 2^{-r} \pi$.  Then we construct the set of integers $\gamma_j = 2^{-t} (\alpha' \ell_j = \alpha \ell'_j)$.  It is a straightforward calculation that $\sum \gamma_j \theta_j = 0$ and $\sum \gamma_j = 0$.  Now consider the expression
\begin{align*}
\sum_{\theta_j \in \esupp{\ket{\mathbf{v}}}} \gamma_j \zeta_j &= \sum_{\theta_j \in \esupp{\ket{\mathbf{v}}}} 2^{-t} (\alpha' \ell_j - \alpha \ell'_j) \zeta_j = (2^{-t} \alpha') \left( \sum_{\theta_j \in \esupp{\ket{\mathbf{v}}}} \ell_j \zeta_j \right) - (2^{-t} \alpha) \left( \sum_{\theta_j \in \esupp{\ket{\mathbf{v}}}} \ell'_j \zeta_j \right).
\end{align*}
We observe that $\sum \gamma_j \zeta_j$ is even, as otherwise we have a contradiction to our hypothesis.  Moreover, $\sum \ell_j \zeta_j$ is odd by our assumption, and at least one of $2^{-t} \alpha$ and $2^{-t} \alpha'$ is odd by definition of $t$.  If $2^{-t} \alpha'$ is odd, then both $2^{-t} \alpha$ and $\sum \ell'_j \zeta_j$ are odd, so $r = s$ and ($\dagger$) is satisfied.  Otherwise, $2^{-t} \alpha'$ is even, so $s > r$, $2^{-t} \alpha$ is odd and $\sum \ell'_j \zeta_j$ is even; thus ($\dagger$) is satisfied.  Therefore, by Kronecker's Theorem (Theorem~\ref{KT}), the system of inequalities
\[
| \theta_j \tau - (\delta + \zeta_j \pi)| < \epsilon \pmod{2 \pi}, \qquad (\theta_j \in \esupp{\ket{\mathbf{v}}})
\]
admits a solution $\tau_0$ for $\tau$.  Hence we obtain
\[
U(\tau_0) = \sum_{\theta_j \in \esupp{\ket{\mathbf{v}}}} e^{i \tau_0 \theta_j} E_j = \sum_{\theta_j \in \esupp{\ket{\mathbf{v}}}} (1 - \epsilon') e^{i \delta} \sigma_j E_j
\]
and so $U(\tau_0) \ket{\mathbf{v}} \approx e^{i \delta} \ket{\mathbf{v}^\sigma}$, and hence we have pretty good state transfer between states $\ket{\mathbf{v}}$ and $\ket{\mathbf{v}^\sigma}$.

Conversely, suppose that pretty good state transfer occurs between $\ket{\mathbf{v}}$ and $\ket{\mathbf{v}^\sigma}$.  Applying Lemma~\ref{lem:path-sc}, we see that for some $\tau$, we have
\[
\begin{array}{rclcl}
U(\tau) \ket{\mathbf{v}} &\approx& e^{i \delta} \ket{\mathbf{v}^\sigma}, \\
e^{i \theta_j \tau} E_j \ket{\mathbf{v}} &\approx& e^{i \delta} E_j \ket{\mathbf{v}^\sigma}, &\qquad& (\theta_j \in \esupp{\ket{\mathbf{v}}}), \\
e^{i \theta_j \tau} E_j \ket{\mathbf{v}} &\approx& (-1)^{j + 1} e^{i \delta} E_j \ket{\mathbf{v}} &\qquad& (\theta_j \in \esupp{\ket{\mathbf{v}}}), \\
\theta_j &\approx& \delta + \zeta_j \pi \pmod{2 \pi}, &\qquad& (\theta_j \in \esupp{\ket{\mathbf{v}}}).
\end{array}
\]
So by Kronecker's Theorem (Theorem~\ref{KT}), for every set of integers $\{\ell_j\}$ such that $\ell_j \theta_j = 0$, we have
\[
\sum_{\theta_j \in \esupp{\ket{\mathbf{v}}}} \ell_j (\delta + \zeta_j \pi) \equiv 0 \pmod{2 \pi}.
\]
It follows that if $\sum \ell_j \zeta_j$ is odd, then $\sum \ell_j$ cannot be zero, or the above condition is not satisfied, which completes the proof.
\end{proof}

Next, we provide our key lemma, which uses cyclotomic polynomials to draw conclusions about the possible linear combinations of eigenvalues that equal zero, which will aid us in applying Kronecker's Theorem to derive our main results.

\begin{lemma} \label{lem:cyclo}
Let $m$ be a positive integer of the form $2^t p^s$, where $p$ is an odd prime and $s, t \in \N$, and let $\theta_j = 2 \cos (j \pi / m)$, $1 \le j < m$.  If there is a linear combination satisfying
\[
\sum_{j = 1}^{m - 1} \ell_j \theta_j = 0,
\]
where each $\ell_j$ is an integer, then if $1 \le j \le m - m / p$, and we let $j := q (m / p) + r$, $0 \le r < m /p$, we have
\[
\ell_j = \begin{cases}
\ell_{m - j} + (-1)^q (\ell_{m - m/p + r} - \ell_{m/p - r}), & r \neq 0; \\
\ell_{m - j}, & r = 0.
\end{cases}
\]
\end{lemma}

\begin{proof}
Notice that each $\theta_j$ is of the form $\theta_j = \zeta_{2m}^j + \zeta_{2m}^{-j}$, where $\zeta_{2m}$ is a $2m$-th root of unity.  Hence, every $\theta_j$ belongs to the cyclotomic field $\Q(\zeta_{2m})$.  The cyclotomic polynomial is
\[
\Phi_{2m}(x) = \sum_{k = 0}^{p - 1} (-1)^k x^{k m / p}
\]
and we define the polynomial $P(x)$ as follows:
\[
P(x) = \sum_{j = 1}^{m - 1} \ell_j x^j + \sum_{m + 1}^{2m - 1} \ell_{2m - j} x^j.
\]
We see that $\zeta_{2m}$ is a root of $P(x)$ and, since $\Phi_{2m}(x)$ is the minimal polynomial of $\zeta_{2m}$, we see that $\Phi_{2m}(x)$ divides $P(x)$.

Let $Q(x)$ be the following polynomial:
\begin{align*}
Q(x) = \sum_{j = 1}^{m / p} \ell_j x^j &+ \sum_{j = m/p + 1}^{m - 1} (\ell_j + \ell_{j - m/p}) x_j + \ell_{m - m / p} x^m + \sum_{j = 1}^{m / p - 1} (\ell_{m - j} + \ell_{m - m / p + j} - \ell_j) x^{m + j}.
\end{align*}
Now, as the degree of $Q(x)$ is $m + m / p - 1$, and
\[
[x^j] \Phi_{2m}(x) Q(x) = [x^j] P(x), \quad 0 \le j \le m + m / p - 1,
\]
we conclude that $Q(x)$ is the unique such polynomial, and since the quotient $P(x) / \Phi_{2m}(x)$ also has this property, it follows that $P(x) = \Phi_{2m}(x) Q(x)$.  Hence, from the coefficients of $x^{2m - j}$ for $1 \le j \le m - m / p$, we have
\[
\ell_j = \begin{cases}
\ell_{m - j} + (-1)^q (\ell_{m - m/p + r} - \ell_{m / p - r}), & r \neq 0; \\
\ell_{m - j}, & r = 0,
\end{cases}
\]
as desired.
\end{proof}

If $\ket{\mathbf{v}}$ is such that $\beta_x = 0$ either for all even $x$ or for all odd $x$, we say $\ket{\mathbf{v}}$ is a \emph{parity state}.  We demonstrate a symmetry property for the eigenvalue support of parity states of paths.  

\begin{lemma} \label{lem:parity}
For $P_n$, let $\ket{\mathbf{v}}$ be a parity state.  Then $\theta_j \in \esupp{\ket{\mathbf{v}}}$ if and only if $\theta_{n + 1 - j} \notin \esupp{\ket{\mathbf{v}}}$.
\end{lemma}

\begin{proof}
It suffices to prove for all $j$ that if $\theta_j \notin \esupp{\ket{\mathbf{v}}}$, then $\theta_{n + 1 - j} \notin \esupp{\ket{\mathbf{v}}}$.  Let $\theta_j \notin \esupp{\ket{\mathbf{v}}}$.  Then for each $x \in V(P_n)$, we have
\[
\bra{x} E_j \ket{\mathbf{v}} = \frac{2}{n + 1} \left( \sin \frac{x j \pi}{n + 1} \right) \sum_{y \in V(P_n)} \beta_y \left( \sin \frac{y j \pi}{n + 1} \right) = 0.
\]
Thus
\begin{align*}
\bra{x} E_{n + 1 - j} \ket{\mathbf{v}} &= \frac{2}{n + 1} \left( \sin \frac{x (n + 1 - j) \pi}{n + 1} \right) \sum_{y \in V(P_n)}\beta_y \left( \sin \frac{y (n + 1 - j) \pi}{n + 1} \right) \\
&= \frac{2}{n + 1} \left( -\cos(x \pi) \sin \frac{x j \pi}{n + 1} \right) \sum_{y \in V(P_n)} \beta_y \left(-\cos(y \pi) \sin \frac{y j \pi}{n + 1} \right) \\
&= \pm \frac{2}{n + 1} \left( \sin \frac{x j \pi}{n + 1} \right) \sum_{y \in V(P_n)} \beta_y \left( \sin \frac{y j \pi}{n + 1} \right) = 0.
\end{align*}
Hence, $\theta_{n + 1 - j} \notin \esupp{\ket{\mathbf{v}}}$ as desired.
\end{proof}

Finally, we will make use of the following trigonometric identity to prove our main results; a derivation can be found in~\cite{vB19}.

\begin{lemma} \label{lem:trig}
Let $n = km$, where $k$ is a positive integer and $m > 1$ is an odd integer, and $0 \le a < k$ be an integer.  Then
\[
\sum_{j = 0}^{m - 1} (-1)^j \cos \left( \frac{ (a + jk)\pi}{n} \right) = 0.
\]
\end{lemma}

We are now able to derive our main results, relating pretty good state transfer on paths to the eigenvalue support of the initial state.

\begin{theorem} \label{thm-odd}
Suppose $m = 2^t p^s$, where $p$ is an odd prime and $s, t \in \N$, and let $\ket{\mathbf{v}}$ be a parity state of $P_{m - 1}$.  For $1 \le c < m / p$, let
\[
S_c := \{\theta_{c + j m / p} : 0 \le j < p\}.
\]
Moreover, let $S_0 := \{\theta_{m / 2}\}$.  Then in $P_{m - 1}$ there is pretty good state transfer between states $\ket{\mathbf{v}}$ and $\ket{\mathbf{v}^\sigma}$ if and only if there does not exist $S_c$ with $c$ odd and $S_{c'}$ with $c'$ even such that $S_c \cup S_{c'} \subseteq \esupp{\ket{\mathbf{v}}}$.
\end{theorem}

\begin{proof}
First, suppose there exist $S_c$ with $c$ odd and $S_{c'}$ with $c'$ even such that $S_c \cup S_{c'} \subseteq \esupp{\ket{\mathbf{v}}}$.  Consider the set of integers $\{\ell_k\}$ given by
\[
\ell_k = \begin{cases}
1, & \mbox{if } k \equiv c, c' + m / p \pmod{2m / p}, c' \neq 0; \\
-1, & \mbox{if } k \equiv c', c + m / p \pmod{2m / p}, c' \neq 0 \mbox{ or } k = m / 2, c' = 0; \\
0, & \mbox{otherwise}.
\end{cases}
\]
For $S_0$, we have that $\theta_{m / 2} = 0$.  Otherwise, by Lemma~\ref{lem:trig}, we have that
\[
\sum_{j = 0}^{p - 1} (-1)^j \theta_{c + jm / p} = \sum_{j = 0}^{p - 1} (-1)^j \cos \left( \frac{(c + jm / p) \pi}{m} \right) = 0,
\]
and so $\sum_k \ell_k \theta_k = 0$.  Moreover, we can verify that $\sum_k \ell_k \zeta_k$ is odd and $\sum_k \ell_k = 0$.  Hence, by Theorem~\ref{thm:KT-path}, we cannot have pretty good state transfer between $\ket{\mathbf{v}}$ and $\ket{\mathbf{v}^\sigma}$.

Now, suppose we do not have $S_c$ with $c$ odd and $S_{c'}$ with $c'$ even such that $S_c \cup S_{c'} \subseteq \esupp{\ket{\mathbf{v}}}$.  Then $S_c \nsubseteq \esupp{\ket{\mathbf{v}}}$ for all odd $c$ or $S_c \nsubseteq \esupp{\ket{\mathbf{v}}}$ for all even $c$.  Consider $S_c \nsubseteq \esupp{\ket{\mathbf{v}}}$, $c \neq 0$.  Then there exists a $j_c$ such that $\theta_{c + j_c m / p} \notin \esupp{\ket{\mathbf{v}}}$, and by Lemma~\ref{lem:parity}, we have that $\theta_{m - c - j_c m / p} \notin \esupp{\ket{\mathbf{v}}}$.  So, in any linear combination, we assume $\ell_{c + j_c m / p} = \ell_{m - c - j_c m / p} \notin \esupp{\ket{\mathbf{v}}}$.  Therefore, letting $r_c \equiv c \pmod{m / p}$, $0 \le r_c < m / p$, we have by Lemma~\ref{lem:cyclo} that $\ell_{m / p - r_c} = \ell_{m - m / p + r_c}$, and hence $\ell_j = \ell_{m - j}$ for every $j \equiv c \pmod{m / p}$.

Now, we first suppose $S_c \nsubseteq \esupp{\ket{\mathbf{v}}}$ for all odd $c$.  Then it follows that $\ell_j = \ell_{m - j}$ for all odd $j$.  Now suppose there is a set of integers $\{\ell_j\}$ such that $\sum_{j} \ell_j \theta_j = 0$ and $\sum_{j} \ell_j \zeta_j$ is odd.  Then since the sum of the $\ell_j$'s for $j$ odd is even, it follows that $\sum_{j} \ell_j \neq 0$.  Hence, by Theorem~\ref{thm:KT-path}, there is pretty good state transfer between $\ket{\mathbf{v}}$ and $\ket{\mathbf{v}^\sigma}$.

Next, we suppose $S_c \nsubseteq \esupp{\ket{\mathbf{v}}}$ for all even $c$.  Then it follows, together with Lemma~\ref{lem:cyclo}, that $\ell_j = \ell_{m - j}$ for all even $j$ and $\ell_{m / 2} = 0$.  Hence $\sum_{j} \ell_j \zeta_j$ is never odd.  Hence, by Theorem~\ref{thm:KT-path}, there is pretty good state transfer between $\ket{\mathbf{v}}$ and $\ket{\mathbf{v}^\sigma}$.
\end{proof}

\begin{theorem}
Suppose $m = p^s$, where $p$ is an odd prime and $s \in \N$, and let $\ket{\mathbf{v}}$ be a parity state of $P_{m - 1}$.  For $1 \le c < m / (2p)$, let
\[
R_c := \{\theta_{c + jm / p} : 0 \le j < p\} \cup \{\theta_{m / p - c + j m / p} : 0 \le j < p\}.
\]
Then in $P_{m - 1}$, there is pretty good state transfer between states $\ket{\mathbf{v}}$ and $\ket{\mathbf{v}^\sigma}$ if and only if there does not exist $R_c$ such that $R_c \subseteq \esupp{\ket{\mathbf{v}}}$.
\end{theorem}

\begin{proof}
First, suppose there exists $R_c$ such that $R_c \subseteq \esupp{\ket{\mathbf{v}}}$.  Consider the set of integers $\{\ell_j\}$ given by
\[
\ell_k = \begin{cases}
1, & \mbox{if } k \equiv c, 2m / p - c \pmod{2m / p}; \\
-1, & \mbox{if } k \equiv c + m / p, m / p - c \pmod{2m / p}; \\
0, & \mbox{otherwise}.
\end{cases}
\]
By Lemma~\ref{lem:trig}, we have that
\begin{align*}
\sum_{j = 0}^{m / p - 1} (-1)^j \theta_{c + j m / p} &= \sum_{j = 0}^{m / p - 1} \cos \left( \frac{(c + jm / p) \pi}{m} \right) = 0, \\
\sum_{j = 0}^{m / p - 1} (-1)^j \theta_{(j + 1) m / p - c} &= \sum_{j = 0}^{m / p - 1} \cos \left( \frac{((j + 1) m / p - c) \pi}{m} \right) = 0,
\end{align*}
and so $\sum_k \ell_k \theta_k = 0$.  Moreover, we can verify that $\sum_k \ell_k \zeta_k$ is odd and $\sum_k \ell_k = 0$.  Hence, by Theorem~\ref{thm:KT-path}, we cannot have pretty good state transfer between $\ket{\mathbf{v}}$ and $\ket{\mathbf{v}^\sigma}$.

Now suppose there does not exist $R_c$ such that $R_c \subseteq \esupp{\ket{\mathbf{v}}}$.  Then for each $c$, there exists a $c'$ such that $\theta_{c'} \in R_c \setminus \esupp{\ket{\mathbf{v}}}$, and by Lemma~\ref{lem:parity}, we have that $\theta_{m - c'} \in R_c \setminus \esupp{\ket{\mathbf{v}}}$.  So, in any linear combination, we assume $\ell_{c'} = \ell_{m - c'} = 0$.  By Lemma~\ref{lem:cyclo}, we have that $\ell_j = \ell_{m - j}$ for every $\theta_k \in R_c$.  It follows, together with Lemma~\ref{lem:cyclo}, that $\ell_j = \ell_{m - j}$ for every $j$.  Now suppose there is a set of integers $\{\ell_j\}$ such that $\sum_j \ell_j \theta_j = 0$ and $\sum_j \ell_j \zeta_j$ is odd.  Then it follows $\sum_j \ell_j \equiv 2 \pmod{4}$, and in particular, is not zero.  Hence, by Theorem~\ref{thm:KT-path}, there is pretty good state transfer between $\ket{\mathbf{v}}$ and $\ket{\mathbf{v}^\sigma}$.
\end{proof}

As a consequence, we demonstrate pretty good state transfer in the following specific family.

\begin{cor}
Given any odd prime $p$ and positive integer $t \ge 2$, there is pretty good state transfer in $P_{2^t p - 1}$ between states $\ket{\mathbf{v}} = \frac{1}{\sqrt{2}} (\ket{a} + \alpha \ket{b})$ and $\ket{\mathbf{v}^\sigma} = \frac{1}{\sqrt{2}} (\ket{2^t p - a} + \alpha \ket{2^t p - b}$ whenever $a \neq b$, $\alpha = \pm 1$, and $a + \alpha b \equiv 0 \pmod{2^t}$.
\end{cor}

\begin{proof}
We consider the eigenvalue support of $\ket{\mathbf{v}}$.  In particular, we show that $\theta_{2p j} \notin \esupp{\ket{\mathbf{v}}}$ for $1 \le j < 2^{t - 1}$.  We have
\begin{align*}
\bra{x} E_{2p j} \ket{\mathbf{v}} &= \frac{1}{\sqrt{2}} \left( \bra{x} E_{2p j} \ket{a} + \alpha \bra{x} E_{2p j} \ket{b} \right) \\
&= \frac{1}{2^{t - 1} p \sqrt{2}} \sin \left( \frac{j x \pi}{2^{t - 1}} \right) \left( \sin \left( \frac{j a \pi}{2^{t - 1}} \right) + \alpha \sin \left( \frac{j b \pi}{2^{t - 1}} \right) \right) \\
&= \frac{1}{2^{t - 2} p \sqrt{2}} \sin \left( \frac{j \pi x}{2^{t - 1}} \right) \sin \left( \frac{ j (a + \alpha b) \pi}{2^t} \right) \cos \left( \frac{j (a - \alpha b) \pi}{2^t} \right) = 0,
\end{align*}
since $a + \alpha b$ is a multiple of $2^t$.  We observe that $2p$ generates the subgroup $\{0, 2, 4, \ldots, 2^t - 2\}$ of $\mathds{Z}_{2^t}$.  Moreover, we have shown that $\theta_{2^{t - 1} p} \notin \esupp{\ket{\mathbf{v}}}$.  Hence, for every even $c$, we have that $S_c \nsubseteq \esupp{\ket{\mathbf{v}}}$, and so by Theorem~\ref{thm-odd}, there is pretty good state transfer between $\ket{\mathbf{v}}$ and $\ket{\mathbf{v}^\sigma}$.
\end{proof}

\section{Concluding Remarks}

Relaxing the requirement of pretty good state transfer on a path to be with respect to a specific state rather than any given state provides a more rich array of possibilities for state transfer than the single state counterpart.  We have described an explicit family of examples involving two qubits, and characterized two families of paths in terms of the eigenvalue support of the initial state.

We would like to complete this characterization of pretty good state transfer of multiple qubit states in terms of the eigenvalue support.  Two of the challenges to overcome are working with the cyclotomic polynomial that results for paths not in this family, and losing the symmetry of the eigenvalue support present for parity states.

Finally, we have focused almost exclusively on transfer to the mirrored state, but other forms of state transfer are also of interest.  Another variation is that of \emph{fractional revival}, in which we start with a single qubit vertex state and desire transfer to a subset of the qubits, including the initial one; see Chan et al.~\cite{Chan2019a} for recent results.  We would like to consider if there is a general problem that captures both of these versions of state transfer.

\bibliographystyle{plain}

\end{document}